\documentclass[10pt,reqno,fleqn]{amsart}
\usepackage[utf8]{inputenc}		%allows even quite old versions of pdflatex to deal with special symbols such as e.g. ä,ü etc.
\usepackage[english]{babel}		%language for automatically created text, e.g. dates or 'contents' vs 'inhaltsverzeichnis'

\usepackage{amsmath}				%more options and improved formatting for mathematical formulas
\usepackage{mathtools}			%extension of amsmath, e.g. \coloneqq, improved equation numbering
\usepackage{amsthm}				%enhances \newtheorem command, recognizes theoremstyle specifications
\usepackage{amssymb}				%commands for a range of symbols e.g. \nrightarrow, \checkmark
\usepackage[nointlimits]{esint}	%LOAD AFTER amslatex(=latex-amsmath & amscls). more integral symbols,e.g. \iiint (triple integral)
\usepackage{mathrsfs} 			%provides special 'formal' script letters via \mathscr{}

\usepackage{graphicx}			%\includegraphics, rotate and scale text
\usepackage{epsfig}				%‘wrapper’ around the graphicx package, can use \psfig{file=xxx,...} for \includegraphics[...]{xxx}
\usepackage{tikz}				%create graphic elements
\usepackage{xxcolor}				%extends xcolor, color, additionally includes environ ment for 'washing out'/'dimming'
\usepackage{wrapfig}
\usepackage{caption}

\usepackage{tabularx}

\usepackage{enumitem}			%allows customization of list environments
\usepackage{ifdraft} 			%for delimdelim (detects whether 'final' or 'draft' chosen for main document -> faster compiling options if draft)

\usepackage[]{todonotes}			%use \todo{} to generate comment boxes
\usepackage[linktocpage=true,colorlinks=true,linkcolor=blue,citecolor=green]{hyperref}
\usepackage[
	backend=biber,style=alphabetic, natbib=false, mcite=false, casechanger=auto,
	sorting=nyt, sortcites=false, pluralothers=true, maxnames=4, minnames=1,
	backref=false, arxiv=abs, 
	date=year,
	isbn=false, url=false, doi=true, eprint=true, related=false,
	giveninits=true
]{biblatex}
\renewbibmacro{in:}{}	%remove the "In:" usually genereated by biblatex in references

\usepackage[hmarginratio=1:1,top=32mm,columnsep=20pt]{geometry} 

\usepackage[normalem]{ulem}

\numberwithin{equation}{section}				%equations are numbered within e.g. 'section', 'subsection'

\theoremstyle{plain}							%STYLE default: italic text, extra space above and below
\newtheorem{theorem}{Theorem}[section]		%[ ]=level on which the counter "theorem" is numbered
\newtheorem*{theorem*}{Theorem}				%*=not number/not counted
\newtheorem*{maintheorem*}{Main Result}
\newtheorem{proposition}[theorem]{Proposition}
\newtheorem{corollary}[theorem]{Corollary}

\theoremstyle{definition}					%STYLE upright text, extra space above and below
\newtheorem{definition}[theorem]{Definition}	
\newtheorem{example}[theorem]{Example}

\theoremstyle{remark}						%STYLE upright text, no extra space above or below
\newtheorem{remark}[theorem]{Remark}

\theoremstyle{definition}					%theoremstyle for placeholderthm
\providecommand{\customtheoremname}{}			%\providecommand = \newcommand but ignored if command already exists
\newtheorem{placeholderthhm}{\customtheoremname}
\newcommand{\newcustomtheorem}[2]{%
  \newenvironment{#1}[1]						%Creates new environment
  {											%This new environment starts with: 						
   \renewcommand\customtheroremname{#2}%		%change the title of placeholder theorem environment
   \renewcommand\theplaceholderthhm{##1}	%set placeholderthhm-counter to argument given when calling this environment
   \placeholderthhm							%\begin{placeholdertheorem}
  }											
  {\endplaceholderthhm}						%This new environment ends with: \end{placeholdertheorem}
}
\newcustomtheorem{customass}{Assumption}		%e.g. \begin{customass}{BBB}: Assumption BBB (not 'Assumption customass-counter')
\newcustomtheorem{customdef}{Definition}

\DeclarePairedDelimiter{\bpair}{[}{]} 			% [] 	square brackets
\DeclarePairedDelimiter{\ppair}{(}{)} 			% ()		parenthesis or round
\DeclarePairedDelimiter{\Bpair}{\{}{\}} 			% {}		braces or curly
\DeclarePairedDelimiterX{\cond}[2]{[}{]}{#1  \,\delimsize\vert\,  \mathopen{}#2}

\newcommand\mathvshift[2]{\raisebox{#1}{$\displaystyle{#2}$}}
\newcommand\strutb{\vphantom{\bigm|}}

\newcommand\bracketinner[2][0pt]{\mathchoice{\mathvshift{#1}{\strutb\mathopen{}#2}}{#2}{#2}{#2}}
\ifoptiondraft{
	\newcommand\pp[2][0pt]{\ppair*{#2}} 		% (
	\newcommand\bp[2][0pt]{\bpair*{#2}} 		% [
	\newcommand\Bp[2][0pt]{\Bpair*{#2}} 		% {
}{
	\newcommand\pp[2][0pt]{\ppair*{\bracketinner[#1]{#2}}} 	% (	pp
	\newcommand\bp[2][0pt]{\bpair*{\bracketinner[#1]{#2}}} 	% [ bp
	\newcommand\Bp[2][0pt]{\Bpair*{\bracketinner[#1]{#2}}}		% { Bp
}
\newcommand\delimstyle[2]{\setlength{\delimitershortfall}{#1}\delimiterfactor=#2}
\delimstyle{5pt}{850}% okay with lmodern

\newcommand{\wt}{\widetilde}
\newcommand{\data}{\mathrm{data}}

\newcommand{\N}{\mathbb{N}}
\newcommand{\Nz}{\mathbb{N}_0}

\newcommand{\R}{\mathbb{R}}

\newcommand{\Ev}{\mathbb{E}}												%might prefer mathsf
\newcommand{\EV}[2][_]{ {\if_#1 \Ev\bp{#2} \else \Ev\cond*{#2}{#1} \fi} }	%\Ev with size adjusted square brackets, optional argument for conditional
\newcommand{\Pm}{{\mathbb{P}}}											%might prefer mathsf
\newcommand{\Qm}{{\mathbb{Q}}}											    %alternative probability measure
\newcommand{\QM}[2][_]{ {\if_#1 \Qm\bp{#2} \else \Pm\cond*{#2}{#1} \fi} }	%\Qm with size adjusted round brackets, optional with condition
\newcommand{\F}{\mathcal{F}}												%sign for sigma-algebra
\newcommand{\pFO}[1][]{ q^{\mathrm{FO}}_{#1} }      %first order termination probability
\newcommand{\pSO}[1][]{ q^{\mathrm{SO}}_{#1} }      %second order termination probability
\newcommand{\KFO}[1][]{ { K^{\mathrm{FO}}_{#1} } }      %first order Kopfschäden
\newcommand{\KSO}[1][]{ { K^{\mathrm{SO}}_{#1} } }      %second order Kopfschäden
\newcommand{\cFO}{ { c^{\mathrm{FO}}_{\mathrm{fixed}} } }      %first order fixed costs
\newcommand{\cSO}{ { c^{\mathrm{SO}}_{\mathrm{fixed}} } }      %second order fixed costs
\newcommand{\rcalc}{ r_{\mathrm{calc}} }              %technical interest rate assumption
\newcommand{\marge}{\alpha}                           %profit margin

\newcommand{\pregross}[2][]{ P^{\mathrm{gross}}_{#2}\pp{#1}}   %gross premium for start age #2 at time #1
\newcommand{\prenet}[2][]{ P^{\mathrm{net}}_{#2}\pp{#1}}   %net premium for start age #2 at time #1
\newcommand{\CFLOW}[1][]{ \mathrm{CF}\pp{#1}}   %net premium for start age #2 at time #1

\newcommand{\aee}[1][]{ { \ddot{a}_{#1} } }
\newcommand{\Aaa}[1][]{ { A_{#1} } }
\newcommand{\RS}[2]{ { {}_{#1}\!V_{#2} } }              %technical provisions for a contract after time #1 for start age #2 

\newcommand{\Cnet}[2]{ { C^{\mathrm{net}}_{#1,#2} } }
\newcommand{\Cgross}[2]{ { C^{\mathrm{gross}}_{#1,#2} } }
\newcommand{\CRS}[2]{ { C^{\RS{}{}}_{#1,#2} } }
\newcommand{\Cfixed}[1]{ { C^{\mathrm{cost}}_{#1} } }

\newcommand{\prenetr}[2][]{ \wt{P}^{\mathrm{net}}_{#2}\pp{#1}}   %net premium for start age #2 at time #1 with real technical interest rate assumption
\newcommand{\RSr}[2]{ { {}_{#1}\!\wt{V}_{#2} } }              %technical provisions for a contract after time #1 for start age #2 with real technical interest rate assumption

\newcommand{\Iinf}[1][]{ { I_{#1} } }                   %inflation index
\newcommand{\Imed}[1][]{ { I^{\mathrm{med}}_{#1} } }    %medical inflation index, #1 = time
\newcommand{\Icost}[1][]{ { I^{\mathrm{cost}}_{#1} } }  %cost inflation index, #1 = time

\newcommand{\BN}[1][]{ { B^{\mathrm{N}}_{#1} } }                   %nominal money market account
\newcommand{\BR}[1][]{ { B^{\mathrm{R}}_{#1} } }                   %real money market account
\newcommand{\FN}[1][]{ { F^{\mathrm{N}}_{#1} } }                   %nominal forward rate
\newcommand{\FR}[1][]{ { F^{\mathrm{R}}_{#1} } }                   %real forward rate
\newcommand{\PN}{ { P^{\mathrm{N}} } }                   %nominal zero coupon prices
\newcommand{\PR}{ { P^{\mathrm{R}} } }                  %real zero coupon prices
\newcommand{\PRmed}{ { P^{\mathrm{R,med}} } }                  %real zero coupon prices
\newcommand{\PRcost}{ { P^{\mathrm{R,cost}} } }                  %real zero coupon prices

\newcommand{\cN}[1][]{ { \lambda^{\mathrm{N}}_{#1} } } 
\newcommand{\cR}[1][]{ { \lambda^{\mathrm{R}}_{#1} } }

\counterwithin{figure}{section}

\begin{document}

% insert [Short Title]{Title}
\title[Market-consistent valuation in health insurance]{On the market-consistent valuation of health insurance liabilities}

% list of authors
\author{Simon Hochgerner}
\address{Financial Market Authority Austria (FMA), Otto-Wagner-Platz 5, 1090 Vienna, Austria}
\email{simon.hochgerner*at*fma.gv.at}

\author{Jonas Ingmanns}
\address{Institute of Science and Technology Austria (ISTA), Am Campus 1, 3400 Klosterneuburg, Austria}
\email{jonas.ingmanns*at*ista.ac.at}

\author{Nicole Kastanek}
\address{Financial Market Authority Austria (FMA), Otto-Wagner-Platz 5, 1090 Vienna, Austria}
\email{nicole.kastanek*at*fma.gv.at}

\begin{abstract}
We are concerned with the market-consistent valuation of lifelong health insurance products, which are subject to adjustments derived from the actuarial equivalence principle and driven by (medical) inflation. Such products are well-established in the European national markets, and the dynamics of the adjustment mechanism is well-understood from an actuarial perspective. However, the question of market-consistent valuation (as is necessary for Solvency II reporting) has not previously been addressed. This gap has led to a situation where some practitioners use stochastic models while others rely on deterministic methods to assign market-consistent values (Best Estimates) to the same type of health insurance liabilities. 

The purpose of this note is to fill this gap by showing that the Best Estimate of a lifelong health insurance policy depends on the choice of model for the interest and inflation rates. That is, the Best Estimate is not uniquely determined by the currently prevailing term structures of nominal and real spot rates, whence a deterministic calculation is theoretically unjustified.

Furthermore, we construct a valuation portfolio such that the Best Estimate valuation decouples into  calculations   of
    1.) deterministic coefficients derived from policy data and 
    2.) the prices of basis financial instruments that are independent of the individual policy data.
Using this decomposition, the policies do not have to be tracked individually along each generated stochastic path. This allows for a more efficient evaluation of the Best Estimate for a large stock of policies with a stochastic model.
\end{abstract}

\thanks{\emph{Disclaimer.} 
The opinions expressed in this article are those of the authors and do not necessarily reflect the official position of the Austrian Financial Market Authority (FMA) }

\maketitle
{\small{\bfseries \href{https://mathscinet.ams.org/mathscinet/msc/msc2020.html}{MSC2020 subject classification}:} 91G30, 91G05}


%%%%%%%%%%%%%%%%%%%%%%%%%%%%%%%%%%%%%%%%%%%%%%%%
%% Table of Contents							%%
%%%%%%%%%%%%%%%%%%%%%%%%%%%%%%%%%%%%%%%%%%%%%%%%
\setcounter{tocdepth}{1}		%no subsections in table of contents
\tableofcontents
\setcounter{tocdepth}{3}		%want subsubsections in pdf navigation menu

%%%%%%%%%%%%%%%%%%%%%%%%%%%%%%%%%%%%%%%%%%%%%%%%
%% Main text entry area						%%
%%%%%%%%%%%%%%%%%%%%%%%%%%%%%%%%%%%%%%%%%%%%%%%%

\section{Introduction}

Under Solvency II, the regulatory insurance framework in the EU since 2016, insurance undertakings are required to carry out a market-consistent valuation of assets and liabilities (full balance sheet approach). This article focuses on liabilities which are associated to health insurance products which have a similar-to-life risk profile (\emph{similar to life techniques (SLT)} in the Solvency~II terminology). The valuation of these products  involves very long-term cash flow projections. Indeed, since SLT health insurance contracts  are lifelong, the time to run-off of such a liability portfolio is  markedly longer than that of a typical life insurance portfolio. 

The health insurance cash-flows are premiums, benefits and expenses. Premiums are calculated according to general actuarial principles, and are a-priori assumed to be constant throughout the contract duration for pricing purposes. 
However, given the long-term nature of health insurance, a specific feature of these contracts is that, in many jurisdictions, premiums can be adjusted to reflect trends in medical inflation, expense inflation, or also the frequency in benefit payout (as affected by morbidity, mortality and surrender). 
The reason for this adjustment mechanism is to ensure that benefits remain at the same relative level, without exposing undertakings to unmanageable inflation risk. 
However, premium adjustments may by significantly higher than the underlying inflation rate due to the underlying actuarial principles. 

The way this adjustment mechanism acts on premiums and reserves (along a given inflation scenario) is well-understood, and a comprehensive review of such products in the Belgian market is presented in \cite{HCDDT19}. 
In \cite{DDHLT17}, the authors study adjustment mechanisms based on variants of the actuarial equivalence principle, and \cite{DGADH17} is concerned with definitions of (hypothetical) surrender values for lifelong health products. 
However, the question of market consistent valuation of such products has, to the best of our knowledge, not been previously addressed in a systematic manner. In this context we note that, while the surrender value definitions suggested in \cite{DGADH17} are formulated in terms of accumulated reserves or previously paid premiums, it is often argued that surrender values should also contain an element of market-consistent valuation in order to make the transaction fair (e.g., via minimum guaranteed surrender values as discussed in \cite{MB16}). 

From a regulatory point of view, the problem of market-consistent valuation arises from the Solvency~II principle of an \emph{arm`s length transaction}, and the specific question of stochastic modelling is raised in Guideline~53A on the Valuation of Technical Provisions \cite{EIOPA22}:   
\begin{quote}
\emph{Insurance and reinsurance undertakings should use stochastic modelling for the valuation of technical provisions of contracts whose cash flows depend on future events and developments, in particular those with material options and guarantees.
}
\end{quote}
To ensure market consistency, the stochastic models are calibrated to current market data.
The expected value of the discounted cash-flows along scenarios sampled from the economic model is called the \emph{Best Estimate} under Solvency~II, and the market consistent value of the liability portfolio is obtained by adding a \emph{Risk Margin} to the Best Estimate. 

In contrast to traditional life insurance, most health insurance products do not have a profit sharing feature, i.e., discretionary benefits that depend on the undertaking's return on investment. 
For this reason and due to the adjustment mechanism, which mitigates the risk undertakings are exposed to, health insurance (without profit participation) is oftentimes viewed as a relatively mundane business.
In particular, insurance undertakings often rely on a purely deterministic approach for the market-consistent valuation, arguing with practical independence from economic developments. 

In the present article, we question this approach. To analyze the relevant effects systematically, we restrict ourselves to the simplest building block of health insurance, namely we assume that the premiums at each point in time are fully determined by the \emph{actuarial equivalence principle} where all adjusted trends are reflected in the claims inflation. 
That is, the actuarial present value of future inflation adjusted claims shall equal the sum of current reserves and the actuarial present value of future premiums, for all contracts and all times. 
This set-up is a simplification over the reality which typically would also involve other forms of trend adjustments and limitations of premium adjustments, such that premiums would not increase beyond a given threshold derived from management rules or regulatory requirements.
 
The cash flow associated to a health insurance policy (without profit participation and without limitation rules) does not explicitly depend on any nominal interest rate scenario, but it is naturally tied to inflation: 
the insurance payouts correspond to medical treatments and medical services.
If each year the cash flow was adjusted exactly according to the inflation rate, still no stochastic model would be necessary because the Best Estimate would be uniquely determined with respect to the current real interest rate curve.
While this is the case for the adjustment of the benefits, the adjustment mechanism for the premiums is more complicated:
these are adjusted according to the actuarial equivalence principle, which depending on the size of the technical provisions can lead to significantly larger adjustments compared to the original inflation rate, see Section~\ref{s_toy_ex} below for a toy model.
 
As our main results, we provide the following insights with regard to Guideline 53A on Valuation of Technical Provisions (loc.\ cit.): 
\begin{enumerate}
    \item 
    In Section~\ref{s_ex}, we show that the Best Estimate for the cash flow associated to even the simplest form of health insurance -- as described above -- is not uniquely determined by the current (nominal and real) interest rate curves, but depends significantly on the joint evolution of nominal and real term structures. 
    This is in contrast to the prevailing folklore that  health insurance  does not necessitate stochastic modeling. 
    \item 
    In Section~\ref{s_BE}, we construct a \emph{valuation portfolio} (see e.g. \cite{WM13}) for such   health insurance policies.
\end{enumerate}
 
The valuation portfolio is introduced by means of a natural decomposition of the Best Estimate cash flows:
these can be written as the linear combination of basis financial instruments with only the deterministic coefficients depending on the individual policy data.
In order to determine the Best Estimate, it is hence not necessary to track each individual policy along the economic scenarios.
Instead, it is sufficient to apply the stochastic model only for the pricing of the basis financial instruments.
This decomposition has a range of advantages:
\begin{enumerate}[label = (\alph*)]
    \item   \textbf{Efficient Monte-Carlo simulations. }
            Rather than computing the premiums and payouts for each policy along each sampled scenario, it is sufficient to compute the coefficients once for each policy and only a fixed number of basis financial instruments has to be evaluated along the scenarios.
    \item   \textbf{Identification of the most sensitive time-periods.}
            The time-periods in which the Best Estimate is most sensitive to interest rate changes are apparent from the coefficients.
    \item   \textbf{Comparability of portfolios and stochastic models.} 
            Different policy portfolios can be easily compared via the summed up coefficients, which are independent of the chosen stochastic model.
            The impact of different stochastic models can be easily measured in terms of the Best Estimates for the basis financial instruments.
\end{enumerate}
Note that the possibility for efficient revaluation of liabilities under different interest rate scenarios is also a prerequisite for an efficient Asset Liability Management (ALM) process.  

The valuation portfolio decomposition holds only under the above assumption that inflation adjustments are not capped, i.e.\ the common industry practice of limitations cannot be captured by this approach.  
In reality, this cap is given by an annual management decision with the implicit goal  of preventing adverse selection processes within the pool of policy holders. 
For modeling purposes, varying limitation rules can be implemented.
The most basic option is to cap the adjustments at a multiple of the inflation rate, but already such a simple rule breaks the linearity and hence the decomposition. 

Finally, we note that this article is about the necessities in relation to the correct and efficient modeling of health insurance cash flows with regard to Best Estimate valuation, it is not about the financial modeling of (nominal or real) yield curves. For the latter we refer to, e.g., \cite{BriMe06}. 

\subsection{Related literature}
In the present paper, our point of departure is, similarly to \cite{DDHLT17}, the actuarial equivalence principle, but we investigate ramifications of the ensuing adjustment mechanism for the Best Estimate calculation rather than studying the adjustment mechanism along a fixed scenario. 

For an overview of the practical process and components used for a market consistent valuation of insurance contracts see \cite{CIA24}.
For an introduction to health insurance similar to life techniques in the German and Austrian style, see \cite{bohn80}.
A major focus has been on properly modeling, estimating and forecasting the underlying quantities such as health claims and surrender rates. A work in this direction which focuses on the German market is \cite{CDLS18}.
In \cite{Hel07, Ch12}, semi-Markov models are used to capture the dependence of surrender rates on the run-time of the policy.
In \cite{CDLS18, Pio20, Pio25}, statistical models to predict the short term development of average health claims are proposed. 
For the US-market, \cite{DGZ18} discusses to which extent different inflation indices are suitable for the adjustment with respect to the adjustment of health claims, and \cite{DH19} contains a similar study for the Belgian market. 
In \cite{Jet18}, a martingale decomposition of the health insurance liabilities is derived with respect to risk-factors using some simplifying assumptions.

\subsection{Structure of the article}
In Section \ref{s_toy_ex}, we introduce a toy model to illustrate the basic structure of cash flows associated to health insurance policies.
In Section \ref{s_model}, we discuss a framework for arbitrage-free inflation modeling based on the foreign-currency analogy between `real' and `nominal' valuations.
In Section \ref{s_ex}, we show that the Best Estimate for the cash flow of the toy model can vary greatly for models chosen within the framework from Section \ref{s_model}. In particular, the Best Estimate is not uniquely determined by the current interest rate curves.
In Section \ref{s_BE}, we construct the basic cashflow of a general health insurance policy and show that it can be written as a linear combination of basis financial instruments with only the coefficients depending on the individual policy data.
In Section \ref{s_con}, we discuss the results and make some related comments.

\section{A toy model for the structure of cash flows in health insurance}\label{s_toy_ex}
In health insurance, an a-priori constant, life-long annual premium is paid by the policyholder who in turn claims variable payouts. 
The premium is determined via the equivalence principle: At each date, the sum of technical provisions and the present value of the premium cash flow must correspond to the present value of the expected payouts with respect to a technical interest rate assumption and the assumed termination probabilities (surrender and death).

As a toy model, consider a policy which is guaranteed to run for exactly three discrete dates with a technical interest rate assumption $\rcalc=0$ and no expected payouts until the final date, say
\begin{equation*}
    K_x = \pp{0,0,30}.
\end{equation*}
According to the equivalence principle and assuming no inflation, the annual premium $P(0)$ determined when closing the policy hence has to satisfy
\begin{equation*}
    3P(0) = 30.
\end{equation*}
However, due to the nature of health insurance, the actual payout is not a nominal value but corresponds to a basket of services and goods and the actual health of the policy holders.
In a modeling framework, this is represented by tying the payouts to an inflation index $(\Iinf[t])_{t=0,1,2}$ with $\Iinf[0]=1$.
Throughout the runtime of the policy, the premiums may be adjusted with the evolution of the inflation index -- based on the equivalence principle.
In practice, there are other factors which can lead to premium adjustments and which can be included in stochastic models, but in the present work we focus only on inflation.

After the first date, the remaining expected payout based on the current observation of the inflation index is given by $(0,30\Iinf[1])$.
Taking into consideration the technical provision derived from the first premium payment, based on the equivalence principle the new annual premium at $t=1$ has to satisfy
\begin{align*}
    \RS{1}{} + 2P(1)  &= 30\Iinf[1]
    &&\text{with}
    &   \RS{1}{} &= P(0).
\end{align*}
At the final date, the expected payout is updated to $30\Iinf[2]$ and, with $P(2)$ added to the technical provision, the final premium $P(2)$ has to satisfy
\begin{align*}
    \RS{2}{} + P(2) &= 30\Iinf[2]
    &&\text{with}
    &   \RS{2}{} &=  P(1) + P(0).
\end{align*}
In summary, the premium cash flow associated to our toy model is given by
\begin{equation}\label{intro_eq_mock-cf}
    P = \pp{10 ,\, 15\Iinf[1]-5 ,\, 30 \Iinf[2] - 15\Iinf[1] - 5}.
\end{equation}
Note that the premium adjustments are larger than the actual inflation rate because, on top of the basic inflation adjustment, they compensates for the relative devaluation or revaluation of the technical provision.

At first glance, it might seem possible to decompose the cash flow into just nominal claims and real, inflation-adjusted claims.
However, note that the formula for $P(2)$ depends not just on the inflation index at time $t=2$, but also at time $t=1$.
As we will see in Section \ref{s_model} and Section \ref{s_ex}, this mismatch is the reason why it is non-trivial to determine the Best Estimate of health insurance policies.

\section{Inflation models based on the foreign-currency analogy}\label{s_model}
In this section, we discuss a framework for arbitrage-free inflation modeling.
We follow the foreign-currency analogy from \cite[Section 2.9 and Chapter 15]{BriMe06} with real values treated as a separate `currency'.
That is, given a fixed set of nominal forward rates $(\FN[t])_{t\in\Nz}$ and real forward rates $(\FR[t])_{t\in\Nz}$, the respective money market accounts $(\BN[t])_{t\in\Nz}$ and $(\BR[t])_{t\in\Nz}$ are given by $\BN[0]=\BR[0]=1$ and
\begin{align*}
        \BN[t]  &\coloneqq  \prod_{s=0}^{t-1} \pp{1+\FN[s]},
    &   \BR[t]  &\coloneqq  \prod_{s=0}^{t-1} \pp{1+\FR[s]} 
\end{align*}
for $t\in\N$. The corresponding inflation index $(\Iinf[t])_{t\in\Nz}$ is defined as the exchange rate
\begin{equation*}
    \Iinf[t]\coloneqq \frac{\BN[t]}{\BR[t]}.
\end{equation*}
In order to obtain an inflation model and a corresponding arbitrage-free pricing tool, one selects stochastic processes $(\BN[t])_{t\in\Nz}, (\BR[t])_{t\in\Nz}\subset\R_{>0}$ with respect to a probability space $(\Omega,\F,\Qm)$ that are adapted to a filtration $(\F_{t-1})_{t\in\Nz}$ with $\F_t\subset\F$ for all $t\in\Nz$ and $\F_{-1}=\F_0=\Bp{\emptyset, \Omega}$.
These processes need to be chosen such that there is no arbitrage for a given economy of securities and their current market prices.
This economy should at least include 
\begin{align*}
        &(\BN[s])_{s\in\Nz} \,\,\pp{\, =\,\, (\Iinf[s]\BR[s])_{s\in\Nz}}, 
    &   &(\PN(s,t))_{s\leq t} \,\,\,\text{and}\,\,\,
        (\Iinf[s]\PR(s,t))_{s\leq t} \quad\text{for all $t\in\Nz$},
\end{align*}
where $\PN(s,t)$ and $\PR(s,t)$ for $s\leq t$ are the respective prices for zero-coupon bonds in each `currency', that is,
\begin{itemize}
    \item   $\PN(s,t)$ is the nominal price at time $s$ for a contract that guarantees its holder one unit of the nominal currency at time $t$, in particular $\PN(t,t)=1$,
    \item   $\PR(s,t)$ is the real price at time $s$ for a contract that guarantees its holder one unit of the real currency at time $t$, in particular $\PR(t,t)=1$.
\end{itemize}
When there is no arbitrage, there exists an equivalent martingale measure for any numeraire, see e.g. \cite[Section 2.1]{BriMe06}.
Without loss of generality, let $\Qm$ be an equivalent martingale measure for the numeraire $(\BN[t])_{t\in\Nz}$.
In particular, $\Qm$ has to satisfy
\begin{align}
    \PN(0,t) &= \EV[\F_0]{\frac{\PN(t,t)}{\BN[t]}} = \EV[\F_0]{\frac{1}{\BN[t]}} = \EV{\frac{1}{\BN[t]}}\label{model_eq_PN0}
\intertext{and}
    \PR(0,t) &= \Iinf[0]\PR(0,t) = \EV[\F_0]{\frac{\Iinf[t]\PR(t,t)}{\BN[t]}} = \EV{\frac{\Iinf[t]}{\BN[t]}} = \EV{\frac{1}{\BR[t]}}\label{model_eq_PR0}
\end{align}
where $\Ev$ is the expected value with respect to $\Qm$ and the current prices $\PN(0,t)$ and $\PR(0,t)$ at $s=0$ are available from market data.

For an arbitrary contingent claim $H\in L^2(\Omega,\F_t,\Qm)$ at time $t\in\Nz$, a no-arbitrage price $\pi$ is given by
\begin{equation*}
    \pi = \EV{\frac{H}{\BN[t]}}.
\end{equation*}
We say that $\pi$ is the Best Estimate for the claim $H$ with respect to $\Qm$.
\begin{remark}[Non-uniqueness]
The price $\pi$ is uniquely determined, if and only if, the claim is attainable with respect to the given economy of securities. 
Otherwise there exists a different equivalent martingale measure that yields a different price, see e.g. \cite{EKQ95, KK96}.
\end{remark}

\begin{example}[The toy model from Section \ref{s_toy_ex}]\label{model_ex_toy}
Assuming such an interest and inflation model has been chosen, the Best Estimate $\pi$ for the current value of the claims connected to the premium cash flow \eqref{intro_eq_mock-cf} for the toy model from Section \ref{s_toy_ex} is given by
\begin{align*}
    \pi &=  -\sum_{t=0}^2 \EV{\frac{P(t)}{\BN[t]}}\\
        &=  -10 \EV{\frac{1}{\BN[0]}} 
            -15 \EV{\frac{\Iinf[1]}{\BN[1]}} + 5 \EV{\frac{1}{\BN[1]}}
            -30 \EV{\frac{\Iinf[2]}{\BN[2]}} + 15 \EV{\frac{\Iinf[1]}{\BN[2]}} + 5 \EV{\frac{1}{\BN[2]}}\\
        &=  -10 - 15\PR(0,1) + 5\PN(0,1) -30\PR(0,2) + 15\EV{\frac{\Iinf[1]}{\BN[2]}} + 5 \PN(0,2),
\end{align*}
where we applied \eqref{model_eq_PN0} and \eqref{model_eq_PR0}.
Note that all components but one are uniquely determined by the current prices of the nominal and real zero-coupon bonds and hence independent of the chosen interest and inflation model.
However, as we will show in the next section, the value of $\EV{\frac{\Iinf[1]}{\BN[2]}}$ does depend on the chosen model and thus is not uniquely determined by just the current prices of the nominal and real zero-coupon bonds.
If this was the case, then it would be sufficient to use a deterministic interest and inflation model corresponding to the forward rates implied by the current prices for the zero-coupon bonds.
\end{example}

\section{The necessity of stochastic interest and inflation rate modeling}\label{s_ex}
In this section, we use the three-date toy model from Section \ref{s_toy_ex} to demonstrate that the best estimate of an insurance contract is not uniquely determined by the current nominal and real spot rates, given in form of the current prices for zero-coupon bonds
\begin{equation*}
    \PN(0,1),\,  \PN(0,2),\,  \PR(0,1),\,  \PR(0,2) \in \R_{>0}.
\end{equation*}
Based on the observations in Example \ref{model_ex_toy}, it is sufficient to show that $\EV{\frac{\Iinf[1]}{\BN[2]}}$ depends on the choice of possible interest rate scenarios and their weights.

\subsection{A deterministic model}
The simplest model is one deterministic scenario. 
In this case, due to \eqref{model_eq_PN0} and \eqref{model_eq_PR0}, the money market accounts have to be chosen as
\begin{align*}
        \BN &= \pp{1,\, \frac{1}{\PN(0,1)},\, \frac{1}{\PN(0,2)}},
    &   \BR &= \pp{1,\, \frac{1}{\PR(0,1)},\, \frac{1}{\PR(0,2)}}.
\end{align*}
We hence obtain
\begin{align*}
    \EV{\frac{\Iinf[1]}{\BN[2]}} = \EV{\frac{\BN[1]}{\BR[1]\BN[2]}} = \frac{\PR(0,1)}{\PN(0,1)}\PN(0,2).
\end{align*}

\subsection{A two-scenario model}
Next, we consider a model with two possible scenarios.
For brevity, we slightly deviate from the standard framework presented in Section \ref{s_model} by allowing $\BN[1]$ and $\BR[1]$ to be random, which corresponds to skipping the first term since usually $\BN[1]$ and $\BR[1]$ are determined by the current forward rates at $t=0$.

For a probability space $\Omega=\Bp{\omega_1,\omega_2}$, the money market accounts are given by 
\begin{align*}
        \BN(\omega_i) &= \pp{1,\, \frac{1}{\cN[i]\PN(0,1)},\, \frac{1}{\PN(0,2)}}     \quad\text{for $i=1,2$},\\
        \BR(\omega_i) &= \pp{1,\, \frac{1}{\cR[i]\PR(0,1)},\, \frac{1}{\cR[i]\PR(0,2)}}     \quad\text{for $i=1,2$},
\end{align*}
with $\cN[1], \cN[2], \cR[1], \cR[2] \in\R_{>0}$.
The corresponding equivalent martingale measure $\Qm$ with $\QM{\omega_i}=p_i\in(0,1)$ for $i=1,2$ is uniquely determined since \eqref{model_eq_PN0} and \eqref{model_eq_PR0} imply that 
\begin{align*}
        p_1+p_2&=1,
    &   \cN[1]p_1 + \cN[2]p_2 &= 1,
    &   \cR[1]p_1 + \cR[2]p_2 &= 1,
\end{align*}
as long as the chosen parameters allow for a solution.
Alternatively, choosing $\cN[1],\cR[1]>0$ and $p_1\in\pp{0,\min\Bp{1, (\cN[1])^{-1}, (\cR[1])^{-1}}}$, the other parameters have to satisfy
\begin{align*}
        p_2     &= 1-p_1,
   &    \cN[2]  &= \frac{1-\cN[1]p_1}{1-p_1},
   &    \cR[2]  &=  \frac{1-\cR[1]p_1}{1-p_1}.
\end{align*}
For this type of model, we obtain
\begin{align*}
    \EV{\frac{\Iinf[1]}{\BN[2]}} 
        = \EV{\frac{\BN[1]}{\BR[1]\BN[2]}}
    &   = \frac{\BN[1]}{\BR[1]\BN[2]}(\omega_1)p_1 + \frac{\BN[1]}{\BR[1]\BN[2]}(\omega_2)p_2\\
    &   =   \frac{\PR(0,1)}{\PN(0,1)}\PN(0,2) \pp{\frac{\cR[1]}{\cN[1]}p_1 + \frac{1-\cR[1]p_1}{1-\cN[1]p_1}(1-p_1)}.
\end{align*}
Varying the parameter choices yields
\begin{align}
    \pp{\frac{\cR[1]}{\cN[1]}p_1 + \frac{1-\cR[1]p_1}{1-\cN[1]p_1}(1-p_1)}  &\rightarrow +\infty    &&\text{ for $\cN[1]\rightarrow 0$},\label{ex_eq_liminfty}\\
    \pp{\frac{\cR[1]}{\cN[1]}p_1 + \frac{1-\cR[1]p_1}{1-\cN[1]p_1}(1-p_1)}  &\rightarrow 0    &&\text{ for $(\cN[1],\cR[1],p_1)\rightarrow\pp{\frac{1}{2},0,1}$}.\label{ex_eq_limzero}
\end{align}
This shows that we can obtain any given price for this cash flow component when we are free to choose the interest and inflation model.
In particular, there is no directional bias compared to the deterministic Best Estimate -- the price could be higher or lower in comparison.

Note that \eqref{ex_eq_liminfty} corresponds to the possibility of a term with extreme inflation followed by a term with negative nominal interest rates.
In contrast, \eqref{ex_eq_limzero} is harder to characterize since $(\cN[1],\cR[1],p_1)\rightarrow\pp{\frac{1}{2},0,1}$ implies $(\cN[2],\cR[2],p_2)\rightarrow\pp{+\infty,+\infty,0}$, but along the path of $\omega_1$ it corresponds to extreme deflation in the first term and no extreme behavior in the second term.

\begin{corollary}
    Let $\PN(0,1),  \PN(0,2),  \PR(0,1),  \PR(0,2) \in \R_{>0}$.
    For any price $\pi>0$ there exist stochastic processes $(\BN[t])_t,(\BR[t])_t\subset\R_{>0}$ with respect to a probability space $(\Omega,\F,\Qm)$ such that 
    \begin{align}
        \PN(0,t) &= \EV{\frac{1}{\BN[t]}},
        &
        \PR(0,t) &=  \EV{\frac{1}{\BR[t]}}\tag{\eqref{model_eq_PN0},\eqref{model_eq_PR0}}
    \end{align}
    and
    \begin{align*}
        \pi = \EV{\frac{\Iinf[1]}{\BN[2]}},
    \end{align*}
    where $\EV{\cdot}$ is the expected value with respect to $\Qm$ and the inflation index is given by $\Iinf[1]\coloneqq \frac{\BN[1]}{\BR[1]}$.
\end{corollary}

\section{The Best Estimate and valuation portfolio of a health insurance policy}\label{s_BE}
In this section, we introduce the structure of the basic cash flow for general health insurance policies, basic meaning without limitations of premium adjustments.
We show that for each year this cash flow can be written as a linear combination of past and present inflation index values with the deterministic coefficients depending on the individual policy data.
Thus, in order to obtain the Best Estimate, it is sufficient to determine prices for the payment corresponding to the value of an inflation index at a fixed time before the maturity. 

\subsection{The basic cashflow of a health insurance policy}\label{subs_BE_cf}
In health insurance, typically each policy is based on and associated with a range of first order parameters, that is
\begin{description}[style=multiline, leftmargin=2.5cm, labelindent=0.5cm, font = \normalfont]
    \item[{$(\KFO[x])_{x\in\Nz}$}]   health benefits, the assumed annual payout to a policyholder of age $x\in\Nz$ at current price levels,
    \item[{$(\pFO[x])_{x\in\Nz}$}]   the assumed probability that a policyholder of age $x\in\Nz$ dies or surrenders within the next year,
    \item[{$\rcalc\in\R$}]           the technical interest rate assumption,
    \item[{$\cFO\geq 0$}]            the assumed annual fixed costs per policy and
    \item[{$\marge \geq 0$}]         the risk or profit margin added to the premiums.
\end{description}
Together with the observed inflation, these quantities are used to determine
\begin{description}[style=multiline, leftmargin=3cm, labelindent=0.5cm, font = \normalfont]
    \item[{$(\pregross[t]{x})_{t\in\Nz}$}] the annual gross premiums at times $t\in\Nz$ for a policy started at time $t=0$ with age $x\in\Nz$,
\end{description}
which we will describe in detail below in Subsection \ref{subs_BE_premium}.
While the first order parameters are chosen conservatively for risk protection, their best estimates, also called second order parameters, are used in order to model the actual cashflow, that is,
\begin{align*}
    &(\KSO[x])_{x\in\Nz},   &&(\pSO[x])_{x\in\Nz},  &&\cSO\geq 0
\end{align*}
are used for health benefits, termination probability, and annual fixed costs.

Given a fixed inflation scenario encoded in terms of 
\begin{description}[style=multiline, leftmargin=2.5cm, labelindent=0.5cm, font = \normalfont]
    \item[{$(\Imed[t])_{t\in\Nz}$}]  the medical inflation index associated to health benefits, $\Imed[0]=1$, 
    \item[{$(\Icost[t])_{t\in\Nz}$}]  the general inflation index associated to fixed costs with $\Icost[0]=1$,
\end{description}
the estimated cash flow $\CFLOW[t]$ at time $t\in\Nz$ for a policy started at $t=0$ with age $x\in\Nz$ is given by 
\begin{align}\label{BE_eq_cfgross}
    \CFLOW[t] \coloneqq \pp{\pregross[t]{x} - \Imed[t]\KSO[x+t] - \Icost[t]\cSO} \prod_{s=0}^{t-1}(1-\pSO[x+s]).
\end{align}
Here, the product estimates the probability that the policy has not been terminated at time $t$.
We implicitly assume that the health benefits are uniformly impacted by medical inflation.

\subsection{Calculating the premium and inflation-based adjustments}\label{subs_BE_premium}
The gross premiums $(\pregross[t]{x})_t$ consist of two parts, the net premiums $(\prenet[t]{x})_t$ covering the health benefits and a second part covering the fixed costs.
The margin $\marge \geq 0$ is factored in on top, that is,
\begin{align}\label{BE_eq_pregross}
    \pregross[t]{x} \coloneqq \pp{\prenet[t]{x} + \Icost[t]\cFO} \frac{1}{1-\marge} 
\end{align}
for $t\in\Nz$. 

The net premium is determined according to the equivalence principle: 
When closing the policy, the premium is chosen such that the present value of an annual payment equals the present value of the health benefits with respect to the technical interest rate assumption and the assumed first order termination probabilities.
After each year, this premium is adjusted based on the observed inflation in a way that the equivalence principle stays intact, now also including the technical provisions that result from the point-wise mismatch between the net premium and the health benefits. 

To be more precise, the present value of an annual payment of one currency unit until the termination of the policy with respect to the technical interest rate assumption $\rcalc$ and the assumed first order termination probabilities $(\pFO[x])_x$ is given by
\begin{equation}\label{BE_eq_aee}
    \aee[x] \coloneqq \sum_{t\geq 0} \prod_{s=0}^{t-1}\frac{1-\pFO[x+s]}{1+\rcalc},
\end{equation}
where $x\in\Nz$ is the current age of the policyholder.
Analogously, the present value of health benefits paid until the termination of the policy and assuming no inflation is
\begin{equation}\label{BE_eq_Aaa}
    \Aaa[x] \coloneqq \sum_{t\geq 0} \pp{\prod_{s=0}^{t-1}\frac{1-\pFO[x+s]}{1+\rcalc}}\KFO[x+s].
\end{equation}
By the equivalence principle, the initial net premium $\prenet[0]{x}$ is therefore chosen as
\begin{align*}
    \prenet[0]{x} &\coloneqq \frac{\Aaa[x]}{\aee[x]} 
    &\pp{\quad\iff \qquad \aee[x]\prenet[0]{x} = \Aaa[x]}.
\end{align*}
This clearly does not imply a point-wise, year-by-year equivalence because, in general, the health benefits are not constant.
Hence, after the first year, the equivalence principle for premiums and health benefits breaks. 
The gap between the two corresponding present values is filled by the technical provision $\RS{1}{x}$, that is, 
\begin{equation*}
    \RS{1}{x} + \aee[x+1]\prenet[0]{x} = \Aaa[x+1],
\end{equation*}
where $\RS{1}{x}$ given by
\begin{equation*}
    \RS{1}{x} = \pp{\prenet[0]{x}-\KFO[x]}\frac{1+\rcalc}{1-\pFO[x]},
\end{equation*}
which is the mismatch between premiums and health benefits with interest according to the technical interest rate assumption and an increase corresponding to the redistribution of provisions from the expected terminations.

However, the health benefits as expected payouts for medical services are tied to the medical inflation index $(\Imed[t])_t$. 
In order to compensate, the net premiums are adjusted each year based on the observed inflation according to the equivalence principle (including the technical provisions) and assuming no further inflation.

\begin{definition}[The inflation adjustment of net premiums]\label{BE_def_prenet}
Let $(\KFO[x])_{x\in\Nz}$ be a health benefit profile, $(\pFO[x])_{x\in\Nz}$ the assumed termination probabilities, and $\rcalc$ the technical interest rate  assumption. 
Given the medical inflation index $(\Imed[t])_t$ and a start age $x_0\in\Nz$, the annual net premiums are given by 
\begin{multline}\label{BE_eq_prenet}
    \prenet[t]{x_0} \coloneqq \frac{\Imed[t]\Aaa[x_0+t]-\RS{t}{x_0}}{\aee[x_0+t]} \\
    \pp{\,\,\iff \quad \RS{t}{x_0} + \aee[x_0+t]\prenet[t]{x_0} = \Imed[t]\Aaa[x_0+t]}
\end{multline}
for $t\in\Nz$ with $(\aee[x])_x$ as in \eqref{BE_eq_aee}, $(\Aaa[x])_x$ as in \eqref{BE_eq_Aaa}, and the technical provisions $\pp{\RS{t}{x_0}}_t$ given by
\begin{align}
        \RS{0}{x_0}     &\coloneqq  0,\notag\\
        \RS{t+1}{x_0}    &\coloneqq  \pp{\RS{t}{x_0}+\prenet[t]{x_0}-\Imed[t]\KFO[x_0+t]}\frac{1+\rcalc}{1-\pFO[x_0+t]}
        \qquad\text{for all $t\in\Nz$.}\label{BE_eq_RS}
\end{align}
\end{definition}

\subsection{The Best Estimate of the health insurance policy cash flow}
Note that by Definition \ref{BE_def_prenet} the annual net premium for each time $t\in\Nz$ is a linear combination of the values $(\Imed[s])_{0\leq s\leq t}$ of the medical inflation index.
This allows for a separation of the stochastic Best Estimate evaluation from the individual policy data.

\begin{proposition}[Decomposition of the net premiums]\label{BE_prop_decomp-net}
Let the annual net premium and the technical provisions be as in Definition \ref{BE_def_prenet}.
Then 
\begin{align}\label{BE_eq_net-decomp}
    \prenet[t]{x} &= \sum_{s=0}^t \Cnet{t}{s} \Imed[s]
    &&\text{and}
    &   \RS{t}{x} &= \sum_{s=0}^{t-1} \CRS{t}{s} \Imed[s]
\end{align}
for all $t\in\Nz$ with the coefficients defined inductively as
\begin{align}
    \Cnet{t}{s}   &\coloneqq      \left\lbrace
                                    \begin{aligned}
                                        \frac{\Aaa[x+t]}{\aee[x+t]}
                                        &=\frac{\sum_{u\geq 0} \pp{\prod_{v=0}^{u-1}\frac{1-\pFO[x+v]}{1+\rcalc}}\KFO[x+t+u]}
                                                {\sum_{u\geq 0} \prod_{v=0}^{u-1}\frac{1-\pFO[x+v]}{1+\rcalc}}       
                                                                                                            &\quad  &\text{for $s=t$},\\
                                         - \frac{\CRS{t}{s}}{\aee[x+t]}
                                        &= - \frac{\CRS{t}{s}}{\sum_{u\geq 0} \prod_{v=0}^{u-1}\frac{1-\pFO[x+v]}{1+\rcalc}}
                                                                                                            &\quad  &\text{else for $s<t$}
                                    \end{aligned}
                                    \right.\label{BE_eq_Cnet}\\
    \CRS{t}{s}      &\coloneqq      \left\lbrace
                                    \begin{aligned}
                                        & \frac{1+\rcalc}{1-\pFO[x+t-1]}\pp{\Cnet{t-1}{t-1}-\KFO[x+t-1]}    &\quad  &\text{for $s=t-1$},\\
                                        & \frac{1+\rcalc}{1-\pFO[x+t-1]}\pp{\CRS{t-1}{s} + \Cnet{t-1}{s}}   &\quad  &\text{else for $s<t-1$}
                                    \end{aligned}
                                    \right.\label{BE_eq_CRS}
\end{align}
for $t\in\Nz$.
\end{proposition}

\begin{proof}
    We prove \eqref{BE_eq_net-decomp} and the formulas for the coefficients by induction. 
    They clearly hold for $t=0$.
    Suppose they hold for some fixed but arbitrary $t\in\Nz$.
    Then \eqref{BE_eq_RS} immediately yields \eqref{BE_eq_net-decomp} for $\RS{t+1}{x}$ with the coefficients given by \eqref{BE_eq_CRS} for $t+1$.
    Using that \eqref{BE_eq_net-decomp} holds for $\RS{t+1}{x}$, we obtain \eqref{BE_eq_net-decomp} for $\prenet[t+1]{x}$ by \eqref{BE_eq_prenet}.
    Plugging the formulas for $\Aaa[x+t]$ from \eqref{BE_eq_Aaa} and for $\aee[x+t]$ from \eqref{BE_eq_aee} into \eqref{BE_eq_prenet} further yields  \eqref{BE_eq_Cnet} for the coefficients.
    This establishes \eqref{BE_eq_net-decomp} and the formulas for the coefficients for $t+1$. 
    By induction, these hence hold for all $t\in\Nz$.
\end{proof}

This decomposition of the net premium implies a similar decomposition for the gross premiums $(\pregross[t]{x})_t$ as defined in \eqref{BE_eq_pregross} and hence for the overall estimated cash flow of the insurance policy at time $t\in\Nz$ given by
\begin{align}
    \CFLOW[t]          &=  \pp{\pregross[t]{x} - \Imed[t]\KSO[x+t] - \Icost[t]\cSO} \prod_{s=0}^{t-1}(1-\pSO[x+s])\tag{\ref{BE_eq_cfgross}}
\intertext{with}
    \pregross[t]{x} &= \pp{\prenet[t]{x} + \Icost[t]\cFO} \frac{1}{1-\marge}. \tag{\ref{BE_eq_pregross}}
\end{align}

\begin{corollary}[Decomposition of the gross premiums and cash flow]\label{BE_cor_decomp-cf}
For the annual net premiums as in Definition \ref{BE_def_prenet}, the gross premiums $(\pregross[t]{x})_t$ as defined in \eqref{BE_eq_pregross} satisfy
\begin{align*}
    \pregross[t]{x} = \sum_{s=0}^t \frac{\Cnet{t}{s}}{1-\marge}\Imed[s]  +  \frac{\cFO}{1-\marge} \Icost[t]
\end{align*}
with $(\Cnet{t}{s})_{s\leq t}$ as in Proposition \ref{BE_prop_decomp-net}.
The overall cash flow from \eqref{BE_eq_cfgross} at time $t\in\Nz$ is hence given by
\begin{align*}
    \CFLOW[t] = \sum_{s=0}^t \Cgross{t}{s} \Imed[s] + \pp{\prod_{s=0}^{t-1}(1-\pSO[x+s])}\pp{\frac{\cFO}{1-\marge}-\cSO}\Icost[t]
\end{align*}
with
\begin{align}
    \Cgross{t}{s}   &\coloneqq      \left\lbrace
                                    \begin{aligned}
                                        &   \pp{\prod_{s=0}^{t-1}(1-\pSO[x+s])}\pp{\frac{\Cnet{t}{t}}{1-\marge} - \KSO[x+t]} &\quad  &\text{for $s=t$},\\
                                        &   \pp{\prod_{s=0}^{t-1}(1-\pSO[x+s])}\frac{\Cnet{t}{s}}{1-\marge}                  &\quad  &\text{else for $s<t$}.
                                    \end{aligned}
                                    \right.\label{BE_eq_Cgross}
\end{align}                                   
\end{corollary}

\begin{corollary}\label{BE_cor_BE}
    Consider a health insurance policy based on 
    \begin{equation*}
        \data \coloneqq \pp{(\KFO[x])_x, (\pFO[x])_x, \cFO, \rcalc, \marge, (\KSO[x])_x, (\pSO[x])_x, \cSO}.
    \end{equation*}
    Let $\Qm$ be an equivalent martingale measure for the given economy with respect to a numeraire $(\BN[t])$ as in Section \ref{s_model}.
    Assuming independence of the estimated health related data from the economy up to the explicitly modeled inflation of health benefits and fixed costs, 
    the Best Estimate $\mathrm{BE}$ at time $t=0$ with respect to $\Qm$ for the cash flow of the health insurance policy closed at time $t=0$ at age $x\in\Nz$ is given by 
    \begin{align}\label{BE_eq_BE-CF}
        \mathrm{BE} = -\sum_{t\geq 0}\pp{\sum_{s=0}^t \Cgross{t}{s} \EV{\frac{\Imed[s]}{\BN[t]}} + \Cfixed{t}\EV{\frac{\Icost[t]}{\BN[t]}}}
    \end{align}
    with $(\Cgross{t}{s})_{s\leq t}= (\Cgross{t}{s}(x,\data))_{s\leq t}$ as in Corollary \ref{BE_cor_decomp-cf} and 
    \begin{align*}
        \Cfixed{t}\coloneqq \pp{\prod_{s=0}^{t-1}(1-\pSO[x+s])}\pp{\frac{\cFO}{1-\marge}-\cSO}
    \end{align*}
    for $t\in\Nz$.
\end{corollary}

\begin{remark}[Best Estimate for running policies]
For the Best Estimate of a policy that started at some point in the past we obtain the same structural decomposition as in \eqref{BE_eq_BE-CF}.
Based on the available data associated to the policy, there are two options to obtain the respective coefficients:
    \begin{enumerate}
        \item   If the data for a contract closed at time $t=t_0<0$ are given in the same way as for a new policy, the coefficients $(\Cgross{t}{s})_{t\geq s\geq t_0}$ and $(\Cfixed{t})_{t\geq t_0}$ can be calculated as before and for $\Imed[s]$ with $s<0$ we can plug in deterministic historical inflation data.
        \item   Alternatively, if the current premium is known, independent of changes to $\data$ in the past, the corresponding technical provisions are uniquely determined. 
                Adding these technical provisions to the first order health benefits at time $t=0$ and then calculating the coefficients as for an otherwise fresh policy yields the correct coefficients.
    \end{enumerate}
\end{remark}

\begin{remark}[Efficiency and application to multiple policies]
    A decomposition of the Best Estimate as in \eqref{BE_eq_BE-CF} has the main advantage that policy data can be processed independently from the inflation scenario in order to obtain $(\Cgross{t}{s})_{t\geq s}$ and $(\Cfixed{t})_t$.
    For example, rather than explicitly computing the individual premiums along sampled scenarios for a Monte-Carlo simulation, it is sufficient to compute or approximate the building blocks
    \begin{equation}\label{BE_eq_bb}
        \EV{\frac{\Imed[s]}{\BN[t]}}     
        \qquad \text{for $t\in\Nz$ and $s\in\Bp{1,\ldots, t-1}$.}
    \end{equation}
    The remaining building blocks 
    \begin{equation*}
        \EV{\frac{1}{\BN[t]}}=\PN(0,t) ,\quad \EV{\frac{\Imed[t]}{\BN[t]}}=\PRmed(0,t),\quad \EV{\frac{\Icost[t]}{\BN[t]}}=\PRcost(0,t) 
       % \qquad \text{for $t\in\Nz$}
    \end{equation*}
    for $t\in\Nz$ correspond to the prices of nominal and the respective real zero-coupon bonds and hence are usually read off from market data up to a potential adjustment for medical versus core inflation.
    
    In particular, for the Best Estimate of a large portfolio of policies we obtain the same structure as in \eqref{BE_eq_BE-CF} by summing up the respective coefficients for all policies.
    This reduces the computational complexity and allows the usage of a much larger sample of economic scenarios to approximate the building blocks in \eqref{BE_eq_bb} by Monte-Carlo simulations.
\end{remark}

\begin{remark}[Construction of the basis financial instruments from standard financial instruments]
    One way to interpret the basis financial instruments priced in \eqref{BE_eq_bb} is as the combination of a real-valued zero-coupon bond with maturity $s$ and a receiver interest rate swap exchanging the nominal interest rates with a fixed zero-interest rate from $s$ to $t$.
    Individually, prices for both of these products should be available based on current market data. 
    However, the pricing is more complicated since the nominal value for the interest rate swap is given by the random value of the inflation index at time $s$, which is the payout of the real-valued zero-coupon bond.
    In general, the price of an interest rate swap is not independent from the value of the inflation index at time $s$.
\end{remark}

\begin{remark}[Limitation of premium adjustments]\label{BE_rem_lim}
It is a common practice among insurance providers to limit the adjustment of premiums if the adjustment would be extraordinarily large both in absolute terms and relative to the core inflation.
On the level of the premium calculation, each limitation is accompanied by a compensation added to the technical provision.

However, we ignore this practice for the main part of this article. 
In particular, it would break the decomposition of cash flows as in Corollary \ref{BE_cor_decomp-cf} and would hence again require the individual tracking of policies along inflation scenarios.
Nevertheless, this decomposition encodes key properties of the underlying policy portfolio and the stochastic model.
See Section \ref{s_con} for further comments on the impact of limitations.
\end{remark}

\section{Discussion and Conclusion}\label{s_con}
We have shown that it is necessary to choose and apply a suitable stochastic interest rate model in order to reasonably evaluate the Best Estimate for cash flow structures found in health insurance policies.
We further observed that the cash flow associated to basic health insurance policies in a given year can be written as a linear combination of past to present values of an inflation index.
This decomposition allows to obtain the Best Estimate by two separate calculation steps: 
A deterministic computation of the coefficients based on a given portfolio of policies and the application of a stochastic model to determine arbitrage-free prices of the basis financial instruments given by delayed payouts corresponding to the value of an inflation index before the actual maturity.
Due to this decomposition, the cost of a Monte-Carlo simulation for the Best Estimate is essentially independent of the policy portfolio size.
In addition, the structure of different policy portfolios can be easily compared in terms of the coefficients.
Similarly, the impact of different stochastic models on the Best Estimate of policy portfolios can be evaluated by comparing the prices for the basis financial instruments.

There are two additional aspects that we want to discuss:
First, we provide additional comments with respect to the industry practice of limiting premium adjustments, which was already mentioned in Remark \ref{BE_rem_lim}.
Second, we discuss an alternative interpretation of the equivalence principle which would lead to premium adjustments exactly according to the inflation rate and hence would make the application of a stochastic model obsolete. 

\subsection{The common practice of limiting premium adjustments}
It is common industry practice to limit premium adjustments to prevent extraordinarily steep increases.
With respect to the equivalence principle, limiting the premium adjustment has to be accompanied by increasing the technical provisions out of the insurance undertaking's funds to compensate for the resulting gap between the expected payouts and premium payments in the future.
While in practice limitations depend on annual management decisions, within the Best Estimate models often an algebraic rule based on simulated surplus is implemented, which for example limits the adjustment if the rate would exceed both an absolute value and a fixed multiple of the core inflation rate when the undertaking's economic situation allows for it.
These kinds of caps introduce a nonlinearity with respect to the values of the inflation indices which breaks the decomposition from Corollary \ref{BE_cor_decomp-cf} and requires that each policy is tracked individually along inflation scenarios.

One reason for these limitations given by insurance undertakings is adverse selection.
They fear that healthy policyholders who on average require less payouts are more likely to surrender the policy in the case of unexpectedly steep premium adjustments, increasing the expected payouts per remaining policyholder.
Note that the impact of steep premium adjustments on surrenders would have to be quite strong for the limitations to be worth it:
Increasing the premiums by one percentage point less than dictated by the equivalence principle -- and thus also decreasing the future inflow of premiums by roughly one percent -- would be beneficial for the insurance undertaking only if for example otherwise at least roughly one percent of policyholders surrender who are perfectly healthy in the sense that they will never require any payouts -- which would also decrease the future inflow of premiums without changing the amount of payouts.

If dynamic termination behavior with this kind of strong adverse selection is included in the model, then the Best Estimate could technically be lower with limitations than without.
In such an environment, the premium adjustment corresponds to an option -- the undertaking can choose an adjustment below the value prescribed by the inflation rate, which it does if it expects the effects of adverse selection to outweigh the full increase.

Outside of this case, limitations strictly increase the Best Estimate. 
Thus, the Best Estimate without limitations, which is obtainable via our decomposition, essentially provides a lower bound for the Best Estimate with limitations.

\subsection{Treating the technical interest rate assumption as a real rate}
Note that the technical interest rate assumption $\rcalc$ from Section \ref{s_BE} is treated both as a nominal and a real interest rate, depending on the context. 
In the equivalence principle as applied in Section \ref{s_BE}, future cash flows are always considered assuming zero inflation, that is, in real terms, see for example \eqref{BE_eq_aee} and \eqref{BE_eq_Aaa}.
From this perspective, the technical interest rate assumption is applied as a real interest rate relative to the medical inflation. 
However, once inflation has been observed from a year $t$ to $t+1$, the expected payouts are adjusted from $\pp{\Imed[t]\KFO[x+s]}_{s\geq t}$ to $\pp{\Imed[t+1]\KFO[x+s]}_{s\geq t+1}$ but the technical interest rate assumption acts as a nominal interest rate on the technical provisions:
\begin{equation}
    \RS{t+1}{x_0}    \coloneqq  \pp{\RS{t}{x_0}+\prenet[t]{x_0}-\Imed[t]\KFO[x_0+t]}\frac{1+\rcalc}{1-\pFO[x_0+t]}.
    \tag{\ref{BE_eq_RS}}
\end{equation}
As an alternative to the industry convention, the technical interest rate assumption could consistently be treated as a real interest rate by instead computing the technical provisions as
\begin{equation*}
    \RSr{t+1}{x_0}    \coloneqq  \pp{\RSr{t}{x_0}+\prenet[t]{x_0}-\Imed[t]\KFO[x_0+t]}\frac{1+\rcalc}{1-\pFO[x_0+t]}\frac{\Imed[t+1]}{\Imed[t]}.
\end{equation*}
In this case, it is straightforward to see that the net premiums from Definition \ref{BE_def_prenet} with $(\RS{t}{x_0})_t$ replaced by $(\RSr{t}{x_0})_t$ would then be given by 
\begin{equation*}
    \prenetr[t]{x_0} = \Imed[t]\prenet[0]{x_0}.
\end{equation*}
This would be a major simplification of the cash flow structure in health insurance with significant consequences:
\begin{itemize}[leftmargin=20pt]
    \item The gross premiums would be adjusted exactly according to the (medical and cost) inflation rate, which is more transparent for the consumer and makes limitations of premium adjustments unnecessary.
    \item For fixed policy data, no stochastic modeling would be necessary since the Best Estimate is uniquely determined by the current interest rate curves.
    \item Since the premium adjustments are less steep, this would lead to higher premiums at the start and lower premiums towards the end of a policy relative to the current method, assuming that the difference between the current `nominal' and the new `real' technical interest rate assumptions corresponds to the expected inflation, see Figure \ref{fig1} for an illustration.
    This structural change would imply the following:
    \begin{itemize}
        \item With respect to the first order termination probabilities, the Best Estimate decreases after this change -- as long as the actual interest rate is higher than the technical interest rate assumption, because later cash flows are discounted heavier than earlier cash flows.
        \item The Best Estimate would be more sensitive to interest rate changes since there is a larger mismatch between the shapes of premiums and payouts.
        \item Early surrenders and deaths would be even worse for policyholders.
    \end{itemize}
\end{itemize}
In summary, this change would simplify the premium adjustments and determination of the Best Estimate calculation, most likely lower the Best Estimate but introduce more volatility and disadvantage demographics with lower life expectancy.

\begin{figure}
    \centering
    \includegraphics[width=0.45\textwidth]{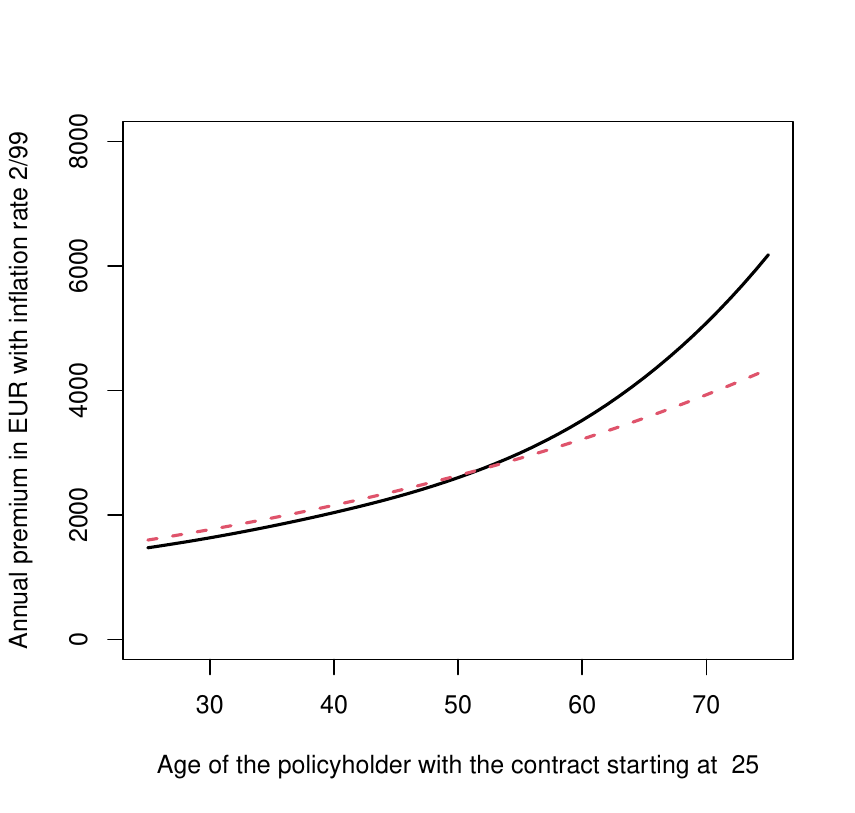}
    $\quad$
    \includegraphics[width=0.45\textwidth]{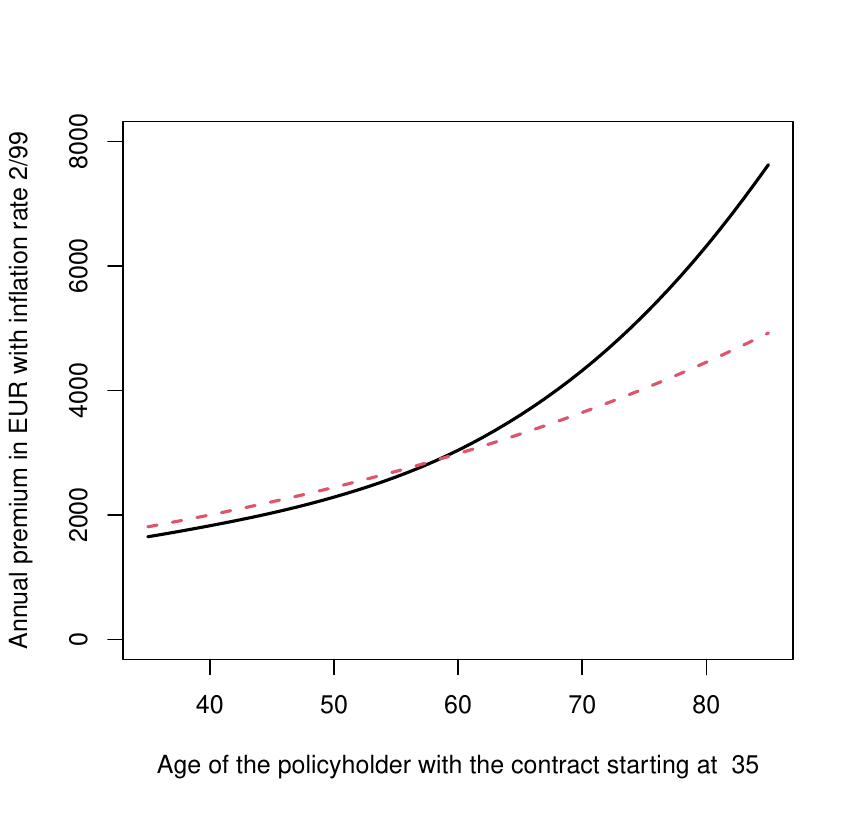}
    \caption{A comparison of the premium development with a nominal versus a real technical interest rate assumption for an inpatient tariff. 
    Based on typical first order health benefits $(\KFO[x])_x$ for an inpatient tariff and assumptions $(\pFO[x])_x$ for the combined probability of death or surrender, the premiums are calculated  with a technical interest rate assumption $\rcalc=1\%$ treated as a nominal rate (depicted in black) and with $\rcalc=-1\%$ treated as a real rate (depicted in red) in a scenario with annual inflation rate $\frac{101}{99}-1 \approx 2\%$. 
    The difference between the initial premiums is less than ten percent. 
    The resulting cash flows have the same present value with respect to the nominal interest rate of $1\%$ and the first order death and surrender assumptions.
    Left: The premium development for a contract starting at age $x_0=25$.
    Right: The premium development for a contract starting at age $x_0=35$.
    }
    \label{fig1}
\end{figure}

\section*{Declarations}
\begin{itemize}
\item Funding: No funds, grants, or other support was received.
\item Competing interests: The authors have no relevant financial or non-financial interests to disclose.
\end{itemize}

%%  Use \cite{...} to cite references in text.             %%
%%%%%%%%%%%%%%%%%%%%%%%%%%%%%%%%%%%%%%%%%%%%%%%%%%%%%%%%%%%%%

\printbibliography[title={References}]

@article{DH19,
  author  = {Dhaene, Jan and Hanbali, Hedia},
  title   = {Measuring medical inflation for health insurance portfolios in Belgium},
  journal = {European Actuarial Journal},
  year    = {2019},
  volume  = {9},
  pages   = {139--153},
  doi     = {10.1007/s13385-019-00200-6},
  url     = {https://doi.org/10.1007/s13385-019-00200-6}
}

@article{MB16,
  author    = {Matthias, Lisa and Balleer, Martin},
  title     = {Zillmerung und Rückkaufswert -- ein kritischer Blick auf die rechtliche Bewertung},
  journal   = {Zeitschrift für die gesamte Versicherungswissenschaft},
  volume    = {105},
  number    = {1},
  pages     = {37--50},
  year      = {2016},
  doi       = {10.1007/s12297-016-0331-4},
  url       = {https://doi.org/10.1007/s12297-016-0331-4}
}

@article{DGADH17, 
title={Lifelong Health Insurance Covers With Surrender Values: Updating Mechanisms In The Presence Of Medical Inflation}, volume={47}, DOI={10.1017/asb.2017.13}, number={3}, journal={ASTIN Bulletin}, 
author={Dhaene, Jan and Godecharle, Els and Antonio, Katrien and Denuit, Michel and Hanbali, Hamza}, year={2017}, pages={803–836}
}

@article{HCDDT19,
title = {Once covered, forever covered: The actuarial challenges of the Belgian private health insurance system},
journal = {Health Policy},
volume = {123},
number = {10},
pages = {970-975},
year = {2019},
issn = {0168-8510},
doi = {https://doi.org/10.1016/j.healthpol.2019.07.005},
url = {https://www.sciencedirect.com/science/article/pii/S0168851019301666},
author = {Hamza Hanbali and Hubert Claassens and Michel Denuit and Jan Dhaene and Julien Trufin}
}

@article{DDHLT17,
  author    = {Denuit, Michel and Dhaene, Jan and Hanbali, Hamza and Lucas, Nathalie and Trufin, Julien},
  title     = {Updating mechanism for lifelong insurance contracts subject to medical inflation},
  journal   = {European Actuarial Journal},
  volume    = {7},
  number    = {1},
  pages     = {133--163},
  year      = {2017},
  doi       = {10.1007/s13385-016-0142-y},
  url       = {https://doi.org/10.1007/s13385-016-0142-y}
}

@book{BriMe06,
  title = {{I}nterest {R}ate {M}odels -- {T}heory and {P}ractice},
  author = {Damiano Brigo and Fabio Mercurio},
  address		= {Berlin Heidelberg New York},
  ISBN = {9783540221494},
  DOI = {10.1007/978-3-540-34604-3},
  publisher = {Springer Berlin, Heidelberg},
  year = {2006}
}

@article{EKQ95,
    author = {El Karoui, Nicole and Quenez, Marie-Claire},
    title = {Dynamic Programming and Pricing of Contingent Claims in an Incomplete Market},
    journal = {SIAM Journal on Control and Optimization},
    volume = {33},
    number = {1},
    pages = {29-66},
    year = {1995},
    doi = {10.1137/S0363012992232579},
    URL = { https://doi.org/10.1137/S0363012992232579}
}

@article{KK96,
 ISSN = {10505164},
 URL = {http://www.jstor.org/stable/2245175},
 author = {Karatzas, Ioannis and S. G. Kou},
 journal = {The Annals of Applied Probability},
 number = {2},
 pages = {321--369},
 publisher = {Institute of Mathematical Statistics},
 title = {On the Pricing of Contingent Claims under Constraints},
 urldate = {2026-03-25},
 volume = {6},
 year = {1996},
 doi = {10.1214/aoap/1034968135}
}

@book{WM13,
    author = {W\"uthrich, Mario V. and Merz, Michael},
    title = {Financial {M}odeling, {A}ctuarial {V}aluation and {S}olvency in {I}nsurance},
    publisher = {Springer},
    address		= {Berlin Heidelberg},
    year = {2013},
    doi = {10.1007/978-3-642-31392-9},
    ISBN = {9783642313912}
}

@misc{CIA24,
  author       = {{Canadian Institute of Actuaries}},
  title        = {{IFRS 17} Market Consistent Valuation of Financial Guarantees for Life and Health Insurance Contracts},
  year         = {2024},
  howpublished = {Educational Note},
  URL = {https://www.cia-ica.ca/publications/224128e/},
    urldate = {2026-04-13},
}

@misc{EIOPA22,
  author       = {{EIOPA}},
  title        = {Revised Guidelines on Valuation of Technical Provisions},
  year         = {2022},
  howpublished = {{EIOPA-BoS-22/217}},
  URL = {https://www.eiopa.europa.eu/publications/revised-guidelines-valuation-technical-provisions_en},
    urldate = {2026-04-15},
}

@article{Pio25,
    author = {Piontkowski, Jens},
    title = {Pricing {G}erman health insurance products with only few insured persons},
    journal = {European Actuarial Journal},
    year = {2025},
    volume = {15},
    number = {3},
    pages = {831-857},
    doi = {10.1007/s13385-025-00427-6}
}

@article{CDLS18, 
    title={Projection models for health expenses}, 
    volume={12}, 
    DOI={10.1017/S1748499517000240}, 
    number={1}, 
    journal={Annals of Actuarial Science}, 
    author={Christiansen, Marcus and Denuit, Michel and Lucas, Nathalie and Schmidt, Jan-Philipp}, 
    year={2018}, 
    pages={185–203}
}

@article{Pio20, 
    title={Forecasting health expenses using a functional data model},
    volume={14},
    DOI={10.1017/S1748499519000046},
    number={1},
    journal={Annals of Actuarial Science},
    author={Piontkowski, Jens},
    year={2020},
    pages={72–82}
}

@article{DGZ18,
    author = {Dunn, Abe and Grosse, Scott D. and Zuvekas, Samuel H.},
    title = {Adjusting Health Expenditures for Inflation: {A} Review of Measures for Health Services Research in the {U}nited {S}tates},
    journal = {Health Services Research},
    year = {2018},
    volume = {53},
    number = {1},
    doi = {10.1111/1475-6773.12612}
}

@article{Ch12,
    author = {Christiansen, Marcus C.},
    title = {Multistate models in health insurance},
    journal = {AStA Advances in Statistical Analysis},
    year = {2012},
    volume = {96},
    number = {2},
    doi = {10.1007/s10182-012-0189-2}
}

@mastersthesis{Hel07,
    author = {Helwich, Marko},
    title={Durational effects and non-smooth semi-Markov models in life insurance},
    year={2007},
    school = {Universit\"at Rostock},
    url={http://rosdok.uni-rostock.de/resolve?urn=urn:nbn:de:gbv:28-diss2008-0056-4}
}

@mastersthesis{Jet18,
    author = {Jetses, Julian},
    title={Durational effects and non-smooth semi-Markov models in life insurance},
    year={2018},
    school = {Carl von Ossietzky Universit\"at Oldenburg},
    url={https://uol.de/f/5/inst/mathe/personen/Jetses_2018_Martingalzerlegung_von_Verbindlichkeiten_in_der_Personenversicherung.pdf},
    urldate = {2026-04-14},
}

@book{bohn80,
  title={Die {M}athematik der deutschen privaten {K}rankenversicherung},
  author={Bohn, Klaus},
  address = {Karlsruhe},
  year={1980},
  publisher={Verlag Versicherungswirtschaft}
}
\end{document}